\documentclass[conference]{IEEEtran}
\IEEEoverridecommandlockouts

\usepackage[utf8]{inputenc} 
\usepackage[T1]{fontenc}
\usepackage{url}
\usepackage{amsmath}
\usepackage{cuted}
\allowdisplaybreaks
\usepackage{amsmath,amssymb,amsfonts}
\usepackage{algorithmicx}
\usepackage{algorithm,algpseudocode}
\usepackage{algcompatible}
\usepackage{array}
\usepackage[caption=false,font=normalsize,labelfont=sf,textfont=sf]{subfig}
\usepackage{textcomp}

\usepackage{url}
\usepackage{verbatim}
\usepackage{graphicx}
\usepackage{cuted}
\usepackage{dblfloatfix} 
\usepackage{mathrsfs}
\usepackage{caption}
\usepackage{cite}
\usepackage{bm}
\usepackage{xcolor}
\usepackage{colortbl}
\usepackage{multirow}
\usepackage{bbm}
\usepackage[figuresright]{rotating}
\allowdisplaybreaks

\newtheorem{definition}{Definition}
\newtheorem{theorem}{Theorem}
\newtheorem{lemma}{Lemma}

\newtheorem{remark}{Remark}
\newtheorem{corollary}{Corollary}
\def\BibTeX{{\rm B\kern-.05em{\sc i\kern-.025em b}\kern-.08em
    T\kern-.1667em\lower.7ex\hbox{E}\kern-.125emX}}

\begin{document}
\title{Straggler-Aware Coded Polynomial Aggregation} 
\author{
\IEEEauthorblockN{Xi Zhong\textsuperscript{1}, Jörg Kliewer\textsuperscript{2} and Mingyue Ji\textsuperscript{1}}

\IEEEauthorblockA{\textit{\textsuperscript{1}Department of Electrical and Computer Engineering}, \textit{University of Florida}, Gainesville, FL, USA\\
Email: \{xi.zhong, mingyueji\}@ufl.edu }

\IEEEauthorblockA{\textit{\textsuperscript{2}Department of Electrical and Computer Engineering}, \textit{New Jersey Institute of Technology}, Newark, NJ, USA\\
Email: jkliewer@njit.edu }}
\maketitle

\begin{abstract}
Coded polynomial aggregation (CPA) in distributed computing systems enables the master to directly recover a weighted aggregation of polynomial computations without individually decoding each term, thereby reducing the number of required worker responses.
However, existing CPA schemes are restricted to an idealized setting in which the system cannot tolerate stragglers.
In this paper, we extend CPA to straggler-aware distributed computing systems with a pre-specified non-straggler pattern, where exact recovery is required for a given collection of admissible non-straggler sets.
Our main results show that exact recovery of the desired aggregation is achievable with fewer worker responses than that required by polynomial codes based on individual decoding, and that feasibility is characterized by the intersection structure of the non-straggler patterns.
In particular, we establish necessary and sufficient conditions for exact recovery in straggler-aware CPA.
We identify an intersection-size threshold that is sufficient to guarantee exact recovery.
When the number of admissible non-straggler sets is sufficiently large, we further show that this threshold is necessary in a generic sense.
We also provide an explicit construction of feasible CPA schemes whenever the intersection size exceeds the derived threshold.
Finally, simulations verify our theoretical results by demonstrating a sharp feasibility transition at the predicted intersection threshold.
\end{abstract}

\section{Introduction}
Distributed computing enables large-scale data processing by decomposing a computation across multiple worker nodes and aggregating their responses at a master node. Coding techniques have emerged as powerful tools to improve the reliability and efficiency of distributed computing systems, particularly in the presence of \emph{stragglers}, which fail to return responses within a reasonable time.

For matrix multiplication and polynomial computation tasks, a large body of prior work applies polynomial-based coding techniques to achieve \emph{exact recovery}, including maximum distance separable (MDS) coded computing \cite{8002642,8006963}, polynomial codes \cite{yu2018polynomialcodesoptimaldesign}, MatDot and PolyDot codes \cite{8765375}, entangled polynomial codes \cite{9174167,8949560}, and their extensions to numerically stable and heterogeneous settings \cite{8849468,9322629}. 
Lagrange coded computing \cite{yu2019lagrangecodedcomputingoptimal} further reduces decoding complexity by encoding sub-computations as evaluations of Lagrange polynomials.
Beyond exact recovery, several works have also studied approximated coded computing for general functions under numerical or probabilistic guarantees \cite{9713954, moradi2024codedcomputingresilientdistributed, 11195640}.

A common feature of most existing coded computing schemes is that they rely on
\emph{individual decoding}, where the master decodes all individual sub-computations before computing the desired result.
Moreover, these schemes typically impose a recovery threshold and are designed
to tolerate \emph{all} straggler sets of a given size.
Once the number of non-straggler workers exceeds the recovery threshold, exact
recovery is guaranteed regardless of which workers return their results.

While robustness to arbitrary straggler sets is sufficient to guarantee exact recovery, it may not be required in practice. In many distributed systems, straggler behavior is not completely arbitrary, and certain non-straggler sets occur with much higher frequency than others. 
This observation motivates relaxing worst-case straggler robustness in coded computing by exploiting statistical or structural regularities in straggler behavior. For example, the authors in \cite{11195640} considered probabilistic straggler models and established recovery guarantees with high probability.
Moreover, for computation tasks involving aggregation, recovering every individual sub-computation via individual decoding introduces unnecessary redundancy. This observation has motivated prior works that target aggregated outputs rather than individual computations.
Gradient coding~\cite{9488815} focuses on recovering the sum of gradients by introducing redundancy across workers.
Another work studies linearly separable computations~\cite{9614153}, where the objective is to exactly recover weighted aggregations of arbitrary computations.
In our parallel work~\cite{XJM2026CPA}, we proposed coded polynomial aggregation (CPA), where the goal is to compute a weighted sum of polynomial computations. By exploiting the aggregation structure and the algebraic properties of polynomials, we characterized the number of worker responses required for exact recovery without individual decoding.

However, the analysis in~\cite{XJM2026CPA} is restricted to an idealized setting in which all workers are assumed to respond, i.e., without straggler tolerance. In practice, stragglers are unavoidable, and the subset of workers that return their results may vary over time.
Moreover, requiring exact recovery under \emph{all} possible straggler sets may be overly conservative, as practical systems often exhibit regularities, where certain workers are more reliable than others.

In this paper, we extend CPA to \emph{straggler-aware} distributed computing systems with a pre-specified non-straggler pattern, where exact recovery is required for a given collection of admissible non-straggler sets. We show that exact recovery of the desired aggregation is achievable with fewer worker responses than that required by polynomial codes based on individual decoding, and that feasibility is characterized by the intersection structure of the pre-specified non-straggler pattern. 

The main contributions are summarized as follows.
\begin{enumerate}
\item We establish necessary and sufficient conditions for exact recovery in straggler-aware CPA.
\item We identify a threshold that guarantees exact recovery on the intersection size of the pre-specified non-straggler pattern. 
Moreover, when the number of admissible non-straggler sets is sufficiently large, we show that this threshold is necessary in a generic sense.
\item We provide explicit CPA schemes that achieve exact recovery, whenever the intersection size of the non-straggler pattern exceeds the derived threshold.
\item Simulations demonstrate a sharp feasibility transition at the predicted intersection threshold, corroborating the theoretical results.
\end{enumerate}

\section{Problem Formulation}
\subsection{CPA over a Pre-Specified Non-Straggler Pattern}
\label{subsec:sa-cpa}
We consider a distributed computing system consisting of a master and a set of $N$ workers, indexed by $[N] =\{0$, $1$, $\ldots$, $N-1\}$. 
Given $K$ data matrices $\boldsymbol{X}_k \in \mathbb{C}^{q \times v}$ for $k \in [K]$, a polynomial function $F(\cdot)$ of degree $d$ that operates element-wise on each data matrix, and a weight vector $\boldsymbol{w} \in \mathbb{C}^{K}$ with $w_k \neq 0$ for $k \in [K]$, the
objective of the system is to compute the weighted aggregation, 
\begin{equation}
\label{eq-goal}
\boldsymbol{Y}
\triangleq
\sum_{k=0}^{K-1} w_k\, F(\boldsymbol{X}_k),
\end{equation}
using the responses from a set of $N-S$ non-straggler workers from a pre-specified non-straggler pattern. Specifically, rather than enforcing recovery for every subset of $N-S$ non-straggler workers, the system is assumed to know a collection of admissible non-straggler sets a priori, and is designed to achieve exact recovery for each such set, which is a subset of $[N]$ with cardinality $N-S$.

The definition of a non-straggler pattern is as follows.
\begin{definition}[Non-Straggler Pattern]
Given positive integers $N$ and $S$, a \emph{non-straggler pattern} is
defined as $\boldsymbol{\mathcal{N}} \triangleq \{\mathcal{N}_g : g \in [G]\}$,
where each $\mathcal{N}_g \subseteq [N]$ is a \emph{non-straggler set} satisfying $|\mathcal{N}_g| = N-S$.
We  define the intersection $\mathcal{I} \triangleq \bigcap_{g \in [G]} \mathcal{N}_g$
and its cardinality $I \triangleq |\mathcal{I}|$.
\hfill$\diamond$
\end{definition}
We further define $L \triangleq N-S-I$.
The parameter $G$ denotes the number of non-straggler sets in the pattern.
For example, $G=1$ corresponds to a single designated non-straggler set, whereas $G=\binom{N}{N-S}$ corresponds to all non-straggler sets.

A CPA scheme over a pre-specified non-straggler pattern $\boldsymbol{\mathcal{N}}$ consists of the following three phases.
\subsubsection{Encoding}
The master selects a set of $K$ distinct \emph{data points} $\{\alpha_k \in \mathbb{C} : k \in [K]\}$, and interpolates an encoder polynomial $E(z)$ such that $E(\alpha_k) =$ $ \boldsymbol{X}_k$ for all $k \in [K]$.
Next, the master selects a set of $N$ distinct \emph{evaluation points} 
$\{\beta_n$ $\in \mathbb{C} :$ $n \in [N]\}$ satisfying
$\{\alpha_k :$ $k \in [K]\}$ $\cap$ $\{\beta_n :$ $n \in [N]\}$ $=$ $\emptyset$.
The master evaluates $E(z)$ at $\{\beta_n: n \in [N]\}$ and sends the coded matrix
$E(\beta_n)$ to worker $n$.

\subsubsection{Computing}
Each worker $n \in [N]$ computes $F(E(\beta_n))$ locally and returns the result to
the master.

\subsubsection{Decoding}
Upon receiving responses from a set of $N-S$ workers
$\mathcal{N} \in \boldsymbol{\mathcal{N}}$, the master interpolates a decoder polynomial $D(z)$  such that
$D(\beta_n) = F(E(\beta_n))$ for all $n \in \mathcal{N}$. The master then evaluates $D(z)$ at the data points
$\{\alpha_k : k \in [K]\}$ and obtains
\begin{equation}
\label{eq-decoded-output}
\widehat{\boldsymbol{Y}}(\mathcal{N})
\triangleq
\sum_{k=0}^{K-1} w_k\, D(\alpha_k).
\end{equation}

We define the feasibility of a CPA scheme over a pre-specified non-straggler pattern $\boldsymbol{\mathcal{N}}$  as follows.
\begin{definition}[Feasibility over a Non-Straggler Pattern]
\label{def-feasible-G}
Fix positive integers $K$, $d$, $S$ and $N$, data points $\{\alpha_k :$  $ k \in [K]\}$ and
evaluation points $\{\beta_n :$ $n \in [N]\}$, and a pre-specified non-straggler pattern $\boldsymbol{\mathcal{N}}$.
A CPA scheme is  \emph{feasible over $\boldsymbol{\mathcal{N}}$}
if  $\widehat{\boldsymbol{Y}}(\mathcal{N}) = \boldsymbol{Y}$
for all $\mathcal{N} \in \boldsymbol{\mathcal{N}}$.
\hfill$\diamond$
\end{definition}

In this paper, we treat the data points $\{\alpha_k : k \in [K]\}$ as fixed system parameters\footnotemark.
\footnotetext{Allowing joint design of the data points $\{\alpha_k : k \in [K]\}$ and the evaluation points $\{\beta_n : n \in [N]\}$ may further enlarge the feasible design space of CPA schemes.}
We assume that they are pairwise distinct, i.e., $\alpha_i \neq \alpha_j$ for all $i \neq j$, and generic\footnotemark in the sense that they are chosen outside a proper algebraic variety determined by the system parameters $K$, $N$, $d$, $S$ and $\boldsymbol{w}$.
\footnotetext{For background on the notion of genericity and algebraic varieties, we refer the reader to \cite{Shafarevich1995}.}

\subsection{Individual Decoding Baseline}
Existing results on polynomial codes~\cite{yu2018polynomialcodesoptimaldesign, 8765375, 9174167, 8949560, 8849468, 9322629, yu2019lagrangecodedcomputingoptimal} recover the desired computation by decoding all individual sub-computations via polynomial interpolation.
Among these works, \emph{Lagrange coded computing}~\cite{yu2019lagrangecodedcomputingoptimal}
serves as a natural baseline for the CPA setting, as it applies to polynomial computation tasks by encoding each data matrix
$\boldsymbol{X}_k$ as an evaluation of an encoder polynomial $E(z)$. 
When applied to the CPA setting, Lagrange coded computing leads to the following decoding strategy.
\begin{definition}[CPA Based on Individual Decoding]
A \emph{CPA scheme based on individual decoding} operates as follows.
Upon receiving responses from a non-straggler set of workers, the master reconstructs all individual evaluations
$F(\boldsymbol{X}_k)$ by interpolating a polynomial $D(z)$ satisfying
$D(\alpha_k)=F(\boldsymbol{X}_k)$ for all $k \in [K]$, and then computes the desired aggregation.
\end{definition}

The following lemma characterizes the minimum number of workers required for
feasibility of CPA schemes based on individual decoding, under \emph{arbitrary}
straggler patterns.
\begin{lemma}
\label{le-local}
For integers $K$, $d$, $S$, and $N$, a CPA scheme based on individual decoding is
feasible under \emph{arbitrary} non-straggler patterns if and only if
$N \ge d(K-1)+S+1$.
\end{lemma}

\begin{IEEEproof}
Under individual decoding, the master interpolates the polynomial $D(z) = F(E(z))$.
Since $\deg(E) \le K-1$ and $F(z)$ has degree $d$, we have $\deg(F(E)) \le d(K-1)$.
Hence, interpolating $F(E(z))$ requires at least $d(K-1)+1$ distinct responses.
To ensure feasibility under arbitrary straggler patterns, the polynomial $F(E(z))$ must be uniquely interpolated from the responses of any non-straggler set of size $N-S$. This requires $N-S \ge d(K-1)+1$.
Conversely, if $N-S \ge d(K-1)+1$, then the responses from any such non-straggler
set provide at least $d(K-1)+1$ distinct evaluations, which uniquely
determines $D(z) = F(E(z))$ and hence all individual values $D(\alpha_k)=F(E(\alpha_k))=F(\boldsymbol{X}_k)$.
\end{IEEEproof}

From Lemma~\ref{le-local}, for a CPA scheme based on individual decoding, guaranteeing exact recovery under arbitrary straggler patterns requires the number of workers to satisfy $N \ge d(K-1)+S+1$.

In this paper, we study CPA under a pre-specified non-straggler pattern $\boldsymbol{\mathcal{N}}$ in the regime $N \le d(K-1)+S$, where exact recovery under \emph{arbitrary} straggler patterns via individual decoding is infeasible.
We show that exact recovery can be achieved by directly exploiting the aggregation structure.

\section{Main Results}
\subsection{Necessary and Sufficient Conditions for Feasibility of CPA over a Pre-Specified Non-Straggler Pattern}
Rather than relying on individual decoding, we directly study the resulting
recovery error $\widehat{\boldsymbol{Y}}(\mathcal{N}_g)-\boldsymbol{Y}$ for
$\mathcal{N}_g \in \boldsymbol{\mathcal{N}}$ in the regime
$N \le d(K-1)+S$.
\begin{theorem}
\label{th-condition}
For positive integers $K$, $d$, $S$, and $N$ satisfying
$S+2 \leq N \leq d(K-1)+S$, let
$C \triangleq d(K-1)+S+1-N$.
For a given pre-specified non-straggler pattern $\boldsymbol{\mathcal{N}}$, a CPA scheme is feasible over $\boldsymbol{\mathcal{N}}$
\emph{if and only if}
the data points $\{\alpha_k \in \mathbb{C} :$ $k \in [K]\}$ and the evaluation points $\{\beta_n \in \mathbb{C} :$ $n \in [N]\}$
satisfy $\{\alpha_k :$ $k \in [K]\}$ $\cap$ $\{\beta_n :$ $n \in [N]\}$ $= \emptyset$
and
\begin{equation}
\label{eq-equations-partial}
\sum_{k=0}^{K-1} w_k\, P_g(\alpha_k)\, \alpha_k^j = 0,
\forall\, j \in [C],\ \forall\, g \in [G],
\end{equation}
where $P_g(z) \triangleq \prod_{n \in \mathcal{N}_g} (z-\beta_n)$, $\mathcal{N}_g \in \boldsymbol{\mathcal{N}}$.
\end{theorem}

\begin{IEEEproof}
The proof is provided in Appendix~\ref{sec-appen-1}.
\end{IEEEproof}
Theorem~\ref{th-condition} shows that a CPA scheme is
feasible over $\boldsymbol{\mathcal{N}}$ if and only if the data and evaluation points
satisfy a system of orthogonality conditions, as given in \eqref{eq-equations-partial}.
In particular, each non-straggler set $\mathcal{N}_g$ induces $C$ orthogonality conditions associated with the polynomial $P_g(z)$.

The intuition behind the proposed conditions in \eqref{eq-equations-partial} is as follows.
The resulting recovery error $\widehat{\boldsymbol{Y}}(\mathcal{N}_{g}) - \boldsymbol{Y}$ can be expressed as a linear combination of 
$\sum_{k=0}^{K-1} w_k P_{g}(\alpha_k)\alpha_k^j$ for $j \in [C]$. When the orthogonality conditions in \eqref{eq-equations-partial} are satisfied,
all such quantities are equal to zero, which eliminates the recovery error and enables exact recovery of $\boldsymbol{Y}$.

\begin{remark}
In Theorem~\ref{th-condition}, the parameter $C$ captures the difference between the number of workers required by individual decoding in Lemma~\ref{le-local} ($d(K-1)+S+1$) and the number of workers required when only the aggregation is recovered ($N$).
\end{remark}

\subsection{A Sufficient Condition Based on the Intersection Structure}
From Theorem~\ref{th-condition}, for fixed data points $\{\alpha_k : k \in [K]\}$ and a given weight vector $\boldsymbol{w}$,
designing a feasible CPA scheme over a $\boldsymbol{\mathcal{N}}$ reduces to selecting evaluation points
$\{\beta_n : n \in [N]\}$ that simultaneously satisfy all orthogonality conditions. 

It can be seen that \eqref{eq-equations-partial} exhibits a common algebraic structure
induced by the intersection of non-straggler sets. Specifically, consider the intersection
$\mathcal{I} \triangleq \bigcap_{g \in [G]} \mathcal{N}_{g}$ and define
$P_{\mathcal{I}}(z) \triangleq \prod_{n \in \mathcal{I}} (z-\beta_n)$.
Since $\mathcal{I} \subseteq \mathcal{N}_g$ for all $g \in [G]$,
the polynomial $P_{\mathcal{I}}(z)$ is a common factor of all
$P_g(z), g \in [G]$. 
Consequently, each orthogonality condition in~\eqref{eq-equations-partial}
can be factorized with respect to $P_{\mathcal{I}}(z)$. This factorization allows
all orthogonality conditions associated with different non-straggler sets
to be \emph{simultaneously enforced} through a reduced set of conditions
that depend only on $P_{\mathcal{I}}(z)$.
Hence, we obtain the sufficient condition stated in Lemma~\ref{le-sufficient}.

\begin{lemma}
\label{le-sufficient}
For positive integers $K$, $d$, $S$, and $N$ satisfying
$S+2$ $\leq N$ $\leq$ $d(K-1)+S$, and a pre-specified non-straggler pattern
$\boldsymbol{\mathcal{N}}$, suppose that
\begin{equation}
\label{eq-equations-partial-equivalent}
\sum_{k=0}^{K-1} w_k\, P_{\mathcal I}(\alpha_k)\, \alpha_k^{j} = 0,\ \  j \in [C+L],
\end{equation}
 where $P_{\mathcal{I}}(z) \triangleq 
\prod_{n \in \mathcal{I}} (z-\beta_n)$.
Then the orthogonality conditions in~\eqref{eq-equations-partial} are satisfied.
\end{lemma}

\begin{IEEEproof}
The proof is provided in Appendix~\ref{appendix-sufficient}.
\end{IEEEproof}
From Lemma~\ref{le-sufficient}, enforcing the orthogonality conditions
in~\eqref{eq-equations-partial} reduces to satisfying
\eqref{eq-equations-partial-equivalent} with respect to the common evaluation
points $\{\beta_n : n \in \mathcal{I}\}$.
Hence, for a fixed set of data points $\{\alpha_k : k \in [K]\}$,
designing a feasible CPA scheme reduces to finding evaluation points
$\{\beta_n : n \in \mathcal{I}\}$.
This reduction motivates the following sufficient lower bound on the intersection
size $I$ for the existence of evaluation points.
\begin{theorem}
\label{th-bound-case1}
For positive integers $K$, $d$, $S$, and $N$ satisfying
$S+2 \le N \le d(K-1)+S$, a pre-specified non-straggler pattern
$\boldsymbol{\mathcal{N}}$, and a fixed generic pairwise distinct set of data
points $\{\alpha_k \in \mathbb{C} : k \in [K]\}$, the following statements hold.
\begin{enumerate}
    \item There exists a choice of $\{\beta_n: n\in \mathcal{I}\}$ such that \eqref{eq-equations-partial-equivalent} holds and
    $\{\beta_n: n\in \mathcal{I}\} \cap \{\alpha_k:k\in [K]\} = \emptyset$ if and only if
    \begin{equation}
    \label{eq-Idk}
    I \ge I^*, \ 
    I^* =
    \begin{cases}
    \left\lfloor \dfrac{K-1}{2} \right\rfloor + 1, & \text{if } d = 1, \\[6pt]
    (d-1)(K-1) + 1, & \text{if } d \ge 2.
    \end{cases}
    \end{equation}
    \item If $I \ge I^*$, then there exists a choice of  $\{\beta_n : n \in [N]\}$ satisfying \eqref{eq-equations-partial} and $\{\alpha_k : k \in [K]\} \cap \{\beta_n : n \in [N]\} = \emptyset$. Hence, the CPA scheme is feasible over $\boldsymbol{\mathcal{N}}$.
\end{enumerate}
\end{theorem}
\begin{IEEEproof}
    The proof is sketched in Appendix \ref{sec-sppen-2}.
\end{IEEEproof}

The first statement in Theorem~\ref{th-bound-case1} characterizes a necessary and sufficient condition on the intersection size $I$, for the existence of evaluation points $\{\beta_n : n \in \mathcal{I}\}$ that satisfy the reduced orthogonality conditions in Lemma~\ref{le-sufficient}.
The second statement shows that the condition $I\geq I^*$ is sufficient to guarantee the existence of evaluation points $\{\beta_n : n \in [N]\}$ satisfying the original orthogonality conditions in~\eqref{eq-equations-partial}, and hence ensures feasibility of the CPA scheme over a non-straggler pattern $\boldsymbol{\mathcal{N}}$.
We refer to $I^*$ as the \emph{sufficient threshold}.

\begin{figure*}
\centering
\includegraphics[width=0.90\textwidth]{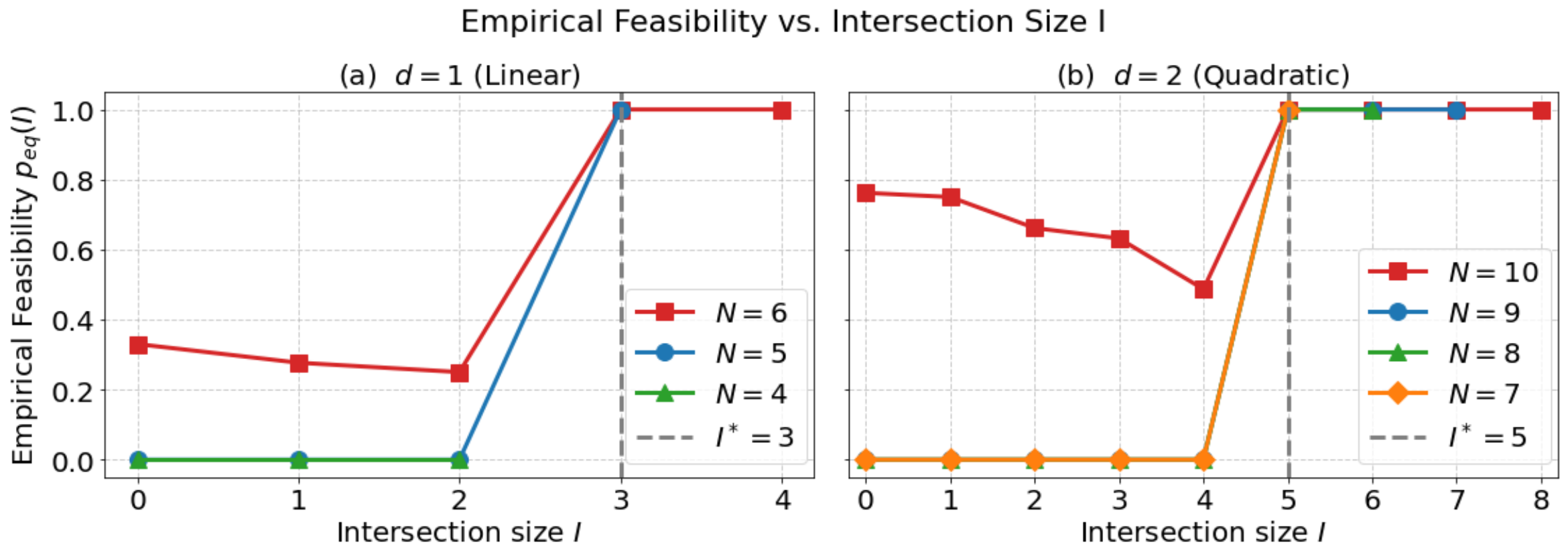}
\caption{Empirical feasibility $p_{eq}(I)$ versus the intersection size $I$ for $K=5$, $S=2$ and $G_{\text{max}} = 7$.
The vertical gray dashed line indicates the sufficient threshold $I^*$ in
Theorem~\ref{th-bound-case1}.
For $d=1$, we consider $N=6,5,4$, corresponding to maximum intersection sizes
$I=4,3,2$, respectively; when $I \le 2$, the curves for $N=5$ and $N=4$ overlap.
For $d=2$, we consider $N=10,9,8,7$, corresponding to maximum intersection sizes
$I=8,7,6,5$, respectively; when $I \le 5$, the curves for $N=9,8,7$ overlap.}
\label{fig:feasibility}
\end{figure*}

The following corollary shows that the sufficient threshold $I^*$ becomes generically necessary when the number of non-straggler sets is sufficiently large.
\begin{corollary}
\label{coro-ne-su-if}
Fix a generic pairwise distinct set of data points $\{\alpha_k: k \in [K]\}$.
Suppose that $G \ge L+1$.
For almost all choices of distinct evaluation points $\{\beta_n:n\in[N]\}$,
i.e., except for a set of Lebesgue measure zero\footnotemark\footnotetext{
The Lebesgue measure-zero set corresponds to a proper algebraic variety.
See Appendix~\ref{appendix-if-sufficient} for details.},
the reduced orthogonality conditions in \eqref{eq-equations-partial-equivalent}
are equivalent to the original orthogonality conditions in~\eqref{eq-equations-partial}.
Consequently, under this regime, the condition $I \ge I^*$ in~\eqref{eq-Idk}
is generically necessary and sufficient for the feasibility of the CPA scheme.
\end{corollary}

\begin{IEEEproof}
The proof is provided in Appendix~\ref{appendix-if-sufficient}.
\end{IEEEproof}
\subsection{Explicit Construction when $I\geq I^*$}
We provide an explicit construction of $\{\beta_n :$ $ n \in [N]\}$ for a $\boldsymbol{\mathcal{N}}$ with $I \ge I^*$, given a fixed generic pairwise distinct $\{\alpha_k:$ $ k \in [K]\}$, such that the resulting CPA scheme is feasible.
The construction adapts Algorithm~1 in~\cite{XJM2026CPA} to the straggler-aware setting by replacing $C$ with $C+L$ and $N$ with $I$.

{\bf Construction of $\{\beta_{n}: n\in [N]\}$:}
 Construct $\boldsymbol{V} \in \mathbb{C}^{(C+L) \times K}$ with $V[j,k]=\alpha_k^j$, $j \in [C+L]$, $k \in [K]$, and $\boldsymbol{A} \in \mathbb{C}^{K \times (I+1)}$ with $A[k,n]=\alpha_k^n$, $n \in [I+1]$, $k \in [K]$.
 Compute $\boldsymbol{U} \triangleq \boldsymbol{V} \operatorname{diag}(\boldsymbol{w}) \boldsymbol{A}$, where $\operatorname{diag}(\boldsymbol{w})$ denotes the diagonal matrix with diagonal $\boldsymbol{w}$.
 Select a non-zero vector $\boldsymbol{c} \in \ker(\boldsymbol{U})$ and define $P_{\text{form}}(z) = \sum_{n=0}^{I} c_n z^n$.
Let $\{\beta_n : n \in \mathcal{I}\}$ be the roots of $P_{\text{form}}(z)$.
For each $g \in [G]$, select distinct values $\{\beta_n : n \in \mathcal{N}_g \setminus \mathcal{I}\}$ from $\mathbb{C}\setminus\{\alpha_k : k \in [K]\}$.
\begin{IEEEproof}
The proof is provided in Appendix~\ref{appendix-if-sufficient}.
\end{IEEEproof}



\section{Simulations}
\label{sec:simulation_method}
In this section, we empirically evaluate the feasibility of CPA as a function of the intersection size $I$.
\subsection{Simulation Setting}
Fixing $K$, $S$, $d$, and $N$, we choose an integer
$1 \leq G_{\max} \leq \binom{N}{N-S}$ and vary
$G \in \{1,2,\ldots,G_{\max}\}$.
For each value of $G$, we uniformly sample $\min\{100, \binom{\binom{N}{N-S}}{G}\}$ distinct instances of the non-straggler pattern $\boldsymbol{\mathcal{N}}$ without replacement.
For each sampled $\boldsymbol{\mathcal{N}}$, we perform the following steps.
Fix  $\{\alpha_k : k \in [K]\}$ as Chebyshev points of the first kind \cite{trefethen2013approximation} on $[-1,1]$, i.e., $\alpha_k = \cos\!\left(\frac{(2k+1)\pi}{2K}\right)$ for $k \in [K]$.
Sample the weight vector $\boldsymbol{w}$ with independent entries,
each drawn uniformly from the interval $(0,1)$.
Numerically test the feasibility of the sampled instance $\boldsymbol{\mathcal{N}}$ by solving for distinct evaluation points
$\{\beta_n : n \in [N]\}$ that satisfy the orthogonality conditions in~\eqref{eq-equations-partial}. 
The approach is nonlinear least squares using
\texttt{scipy.optimize.least\_squares}~\cite{virtanen2020scipy}.
An instance is declared \emph{feasible} if a numerically stable solution satisfying the orthogonality and distinctness conditions is found, and \emph{infeasible} otherwise after $10$ random initializations.

For each intersection size $I$, we quantify the fraction of sampled instances that are numerically feasible.
Specifically, for each $G$, a success rate $p_G(I)$ is defined as the fraction of feasible instances among all sampled instances with intersection size $I$.
The \emph{empirical feasibility} is defined as $p_{\mathrm{eq}}(I) \triangleq \frac{1}{|\mathcal{G}(I)|}
\sum_{G \in \mathcal{G}(I)} p_G(I)$,
where $\mathcal{G}(I)$ denotes the set of values of  $G$  for which at least one sampled instance has intersection size $I$.

\vspace{-0.4\baselineskip}

\subsection{Simulation Results}
We plot the empirical feasibility $p_{\mathrm{eq}}(I) $ as a function of the intersection size
$I$ in Fig.~\ref{fig:feasibility} for both $d=1$ and $d=2$.
From Fig.~\ref{fig:feasibility}, we make the following observations.
For both $d=1$ and $d=2$, the empirical feasibility reaches $100\%$ whenever
$I \ge I^*$. Once the intersection size exceeds the threshold, a set of $\{\beta_n :n\in [N]\}$ satisfying the orthogonality conditions and $\{\beta_n: n\in [N]\} \cap \{\alpha_k: k\in [K]\} =\emptyset$ can be found for the sampled non-straggler pattern, which is consistent with the sufficiency threshold $I^*$ in Theorem~\ref{th-bound-case1}. 
When $I < I^*$, the empirical feasibility drops to $0$ for all cases with
$C \ge 2$, corresponding to $N=5,4$ for $d=1$ and $N=9,8,7$ for $d=2$.
This indicates that no feasible solution is observed among the sampled non-straggler patterns.
When $C=1$,  corresponding to $N=6$ for $d=1$ and $N=10$ for $d=2$, nonzero empirical feasibility is observed when $I<I^*$,
since the number of orthogonality conditions $CG$ becomes comparable to the number of variables $N$, allowing feasible solutions to be found for certain non-straggler patterns.



\appendices

\section{Proof of Theorem~\ref{th-condition}}
\label{sec-appen-1}
We first consider the scalar case, where $\boldsymbol{X}_k$ reduces to a scalar $x_k$, $\boldsymbol{Y}$ reduces to $y$, and $\widehat{\boldsymbol{Y}}(\mathcal{N}_{g})$ reduces to $\widehat{y}(\mathcal{N}_{g})$ for $\mathcal{N}_g \in \boldsymbol{\mathcal{N}}$. The extension to the matrix-valued case follows element-wise, since the polynomial $F(\cdot)$ operates element-wise on the data matrices.

Define the error polynomial $\Delta(z)$ $\triangleq$ $D(z)$ $ - F(E(z))$. From $\deg(D)$ $\le N-S-1$ and $\deg(F(E))$ $\le d(K-1)$, we have 
$\deg(\Delta)$ $\le$ $\max\{\deg(D), \deg(F(E))\}$ $= d(K-1)$.
Then, the recovery error is $\widehat{y}(\mathcal{N}_g) - y$ $=\sum_{k=0}^{K-1}$ $w_k$ $( D(\alpha_k) - F(x_k))$ $=$  $\sum_{k=0}^{K-1}$ $w_k$ $( D(\alpha_k) - F(E(\alpha_k)))$ $= \sum_{k=0}^{K-1} w_k \Delta(\alpha_k)$.

During decoding, $D(\beta_n)=F(E(\beta_n))$ imply $\Delta(\beta_n)$ $=$ $0$ for all $n\in\mathcal{N}_{g}$.
Hence, $\Delta(z)$ admits the factorization $\Delta(z)$ $=$ $\prod_{n\in \mathcal{N}_g}$ $(z-\beta_n)$ $R(z)$ $=$  $P_{g}(z) R(z)$, where $P_g(z) \triangleq$ $\prod_{n \in \mathcal{N}_g} (z-\beta_n)$ and $R(z)$ is a polynomial satisfying $\deg(R) \le$ $d(K-1)- N + S$.
We expand $R(z)=\sum_{j=0}^{\deg(R)} r_j z^j$ with arbitrary coefficients $\{r_j:$ $j \in [\deg(R)+1]\}$. Then,
$\widehat{y}(\mathcal{N}_g)$ $-$ $y = $ $\sum_{k=0}^{K-1} w_k \Delta(\alpha_k) =$ $\sum_{k=0}^{K-1} w_k P_g(\alpha_k)R(\alpha_k)$ $=$ $\sum_{j=0}^{\deg(R)}$ $r_j$ $( \sum_{k=0}^{K-1} w_k P_g(\alpha_k)$ $ \alpha_k^j )$.
Therefore, $\widehat{y}(\mathcal{N}_g)-y=0$ for all admissible choices of $R(z)$ if and only if
$\sum_{k=0}^{K-1} w_k P_g(\alpha_k)\alpha_k^j=0$ for all $j\in[C]$. Enforcing this condition for all $\mathcal{N}_g\in\boldsymbol{\mathcal{N}}$ yields the orthogonality conditions in~\eqref{eq-equations-partial}. The same argument applies element-wise to the matrix-valued case, which completes the proof.

\section{Proof of Lemma \ref{le-sufficient}}
\label{appendix-sufficient}
Define $P_{\mathcal{I}}(z)\triangleq\prod_{n\in\mathcal{I}}(z-\beta_n)$. Since $\mathcal{I}\subseteq\mathcal{N}_g$ for all $g\in[G]$, $P_{\mathcal{I}}(z)$ is a common factor of $P_g(z)$. Thus, we write $P_g(z)=Q_g(z)P_{\mathcal{I}}(z)$, where $Q_g(z)\triangleq\prod_{n\in\mathcal{N}_g\setminus\mathcal{I}}(z-\beta_n)$ satisfies $\deg(Q_g)=N-S-I \triangleq L$.
We expand $Q_g(z)$ $=$ $\sum_{l=0}^{L}$ $q_{g,l}z^{l}$, where $q_{g,l}$ denotes the coefficient in the polynomial $Q_{g}(z)$ and is a function of 
$\{\beta_n: n\in \mathcal{N}_{g} \setminus \mathcal{I}\}$.
Then, each orthogonality condition in \eqref{eq-equations-partial} can be rewritten as
\begin{align}
\sum_{k = 0}^{K-1} w_{k} P_{g}(\alpha_k) \alpha_k^j & = \sum_{k=0}^{K-1} w_k Q_{g}(\alpha_k) P_{\mathcal{I}}(\alpha_k) \alpha_k^j  \notag \\
& =  \sum_{l = 0}^{L} q_{g,l}
 \sum_{k = 0}^{K-1} w_k P_{\mathcal{I}}(\alpha_k) \alpha_{k}^{j+l} = 0, \label{eq-sufficient-I}
\end{align}
 Hence, a sufficient condition for
$\sum_{k = 0}^{K-1}$ $w_{k}$ $P_{g}(\alpha_k)$ $\alpha_k^j $ $= 0$
to hold for all $g \in [G]$ and $j \in [C]$ is that
$\sum_{k = 0}^{K-1}$ $w_k$ $P_{\mathcal{I}}(\alpha_k)$ $\alpha_{k}^{t}$ $= 0$ for $t \in [C+L]$.
The proof of Lemma \ref{le-sufficient} is completed.

\section{Proof Sketch of Theorem~\ref{th-bound-case1}}
\label{sec-sppen-2}
The first statement in Theorem~\ref{th-bound-case1} is equivalent to the following claim.
There exists  $\{\beta_n :$ $n \in \mathcal{I}\}$ satisfying
\begin{align}
&  \sum_{k=0}^{K-1} w_{k}P_{\mathcal{I}}(\alpha_k)\alpha_k^j = 0, j \in [C+L], \label{eq-proof-equations} \\
& \{\alpha_k : k \in [K]\} \cap \{\beta_n : n \in \mathcal{I}\} = \emptyset,
\label{eq-disjoint} \\
&  \beta_n \neq \beta_{n'} \text{ for all } n, n'\in \mathcal{I}  \text{ with }  n \neq n', \label{eq-distinction}
\end{align}
if and only if $I \ge I^*$.
This follows by a direct reduction to the non-straggler setting studied in~\cite{XJM2026CPA}.
Specifically, by Theorem~2 of~\cite{XJM2026CPA} and its proof,
the conditions \eqref{eq-proof-equations}--\eqref{eq-distinction} in our setting have the same algebraic form as
conditions (7)--(9) in Theorem~2 of~\cite{XJM2026CPA}. 
The difference from the non-straggler setting in~\cite{XJM2026CPA} is that we have $C+L$ orthogonality
conditions and a degree-$I$ polynomial $P_{\mathcal I}(z)$. By adapting the proof of Theorem~2
in~\cite{XJM2026CPA}$,$ where $C$ is replaced by $C+L$, the polynomial $P(z)$ is replaced by
$P_{\mathcal I}(z)$, and $N$ is replaced by $I$, it follows that there exist
$\{\beta_n : n \in \mathcal{I}\}$ satisfying \eqref{eq-proof-equations}--\eqref{eq-distinction} 
if and only if $I \ge I^*$. This establishes the first statement.

By Lemma~\ref{le-sufficient}, the set of $\{\beta_n: n \in \mathcal{I}\}$ satisfying conditions \eqref{eq-proof-equations}--\eqref{eq-distinction}
is sufficient to ensure that the orthogonality conditions in Theorem~\ref{th-condition} are satisfied. In addition, $\{\beta_n : n \in \mathcal{N}_g \setminus \mathcal{I}\}$ for $g \in [G]$ can be arbitrarily selected from $\mathbb{C}\setminus\{\alpha_k : k \in [K]\}$, as long as they are distinct. The proof of Theorem~\ref{th-bound-case1} is completed.

\section{Proof of Corollary~\ref{coro-ne-su-if}}
\label{appendix-if-sufficient}
Suppose that $G \ge L+1$.
It suffices to show that if the original orthogonality conditions in~\eqref{eq-equations-partial} hold,
then the reduced conditions in~\eqref{eq-equations-partial-equivalent} also hold.
In~\eqref{eq-sufficient-I}, we represent each orthogonality condition in~\eqref{eq-equations-partial}
by factoring out the common factor $P_{\mathcal{I}}(\alpha_k)$.
Collecting these equations for all $j \in [C]$ and $g \in [G]$, we obtain
$\boldsymbol{Q}\boldsymbol{M}=\boldsymbol{0}$,
where $\boldsymbol{Q}\in\mathbb{C}^{G\times(L+1)}$ has entries
$Q[g,l]=q_{g,l}$ for $g \in [G]$ and $l \in [L+1]$, and
$\boldsymbol{M}\in\mathbb{C}^{(L+1)\times C}$ is defined entry-wise by
$M[l,j]\triangleq \sum_{k=0}^{K-1} w_k\,P_{\mathcal I}(\alpha_k)\,\alpha_k^{j+l}$
for $l\in[L+1]$ and $j\in[C]$.

From $\boldsymbol{Q}\boldsymbol{M}=\boldsymbol{0}$, for each $j\in[C]$, the $j$-th column of
$\boldsymbol{M}$ lies in the null space of $\boldsymbol{Q}$, denoted by
$\operatorname{null}(\boldsymbol{Q})$.
For almost all choices of $\{\beta_n : n \in [N]\}$, i.e., except for a set of Lebesgue measure zero\footnotemark, 
\footnotetext{The set of $\{\beta_n : n \in [N]\}$ for which the matrix $\boldsymbol{Q}$ has column rank strictly less than $L+1$
is described by the zero set of all $(L+1)\times(L+1)$ minors of $\boldsymbol{Q}$.
This set is a proper algebraic variety and therefore has Lebesgue measure zero in $\mathbb{C}^{K}$.}
the matrix $\boldsymbol{Q}$ has full column rank $L+1$. Since $G \ge L+1$, it follows that
$\operatorname{dim}(\operatorname{null}(\boldsymbol{Q}))
= L+1 - \operatorname{rank}(\boldsymbol{Q})
= L+1 - \min(G,L+1)
= 0$.
Hence, the null space of $\boldsymbol{Q}$ is trivial, and therefore
$\boldsymbol{M}=\boldsymbol{0}$.
This implies that
$\sum_{k=0}^{K-1} w_k\,P_{\mathcal I}(\alpha_k)\,\alpha_k^{j+l} = 0$
for all $l\in[L+1]$ and $j\in[C]$.
Equivalently,
$\sum_{k=0}^{K-1} w_k\,P_{\mathcal{I}}(\alpha_k)\,\alpha_k^{j} = 0$
holds for all $j \in [C+L]$, which implies that
\eqref{eq-equations-partial-equivalent} holds.
This establishes the necessity of the condition \eqref{eq-equations-partial-equivalent} in Lemma~\ref{le-sufficient}
and completes the proof.

\section{Proof Sketch of Construction of $\{\beta_n : n \in [N]\}$}
As shown in Appendix~\ref{sec-sppen-2}, the conditions in Lemma~\ref{le-sufficient}
can be reduced to the conditions in Theorem~2 for the non-straggler CPA setting~\cite{XJM2026CPA}.
Since Algorithm~1 in~\cite{XJM2026CPA} is designed to solve Theorem~2 in the non-straggler case,
it can be directly applied to our setting by replacing $C$ with $C+L$ and $N$ with $I$.
As a result, the roots of the resulting polynomial $P_{\text{form}}(z)$ give the desired
$\{\beta_n : n \in \mathcal{I}\}$, which satisfy 
\eqref{eq-equations-partial-equivalent} and
$\{\beta_n : n \in \mathcal{I}\} \cap \{\alpha_k : k \in [K]\} = \emptyset$.
The remaining evaluation points can be selected arbitrarily, as long as they are distinct from
the constructed evaluation points and the given $\{\alpha_k : k \in [K]\}$.

\newpage
\bibliographystyle{IEEEtran}
\bibliography{reference}

\end{document}